\documentclass[runningheads]{llncs}

\usepackage[utf8]{inputenc}
\usepackage[T1]{fontenc}

\usepackage{graphicx}
\usepackage[svgnames]{xcolor}
\usepackage{fca}
\usepackage{mathtools}
\usepackage{marvosym}
\usepackage{cite}
\usepackage[hidelinks]{hyperref}

\usepackage{enumitem}
\usepackage{booktabs}
\usepackage{tikz}
\usepackage{tikzscale}
\usepackage[outline]{contour}   
\contourlength{2.5pt}  
\dimendef\prevdepth=0    
\usepackage{svg}
\usepackage[ruled,vlined]{algorithm2e}
\usepackage{cleveref}
\newcommand{\forlongversion}[1]{}

\newcommand{\todo}[1]{{\bf\color{red} Todo: #1}}

\newcommand{\old}[1]{{\color{teal} #1}}

\renewcommand{\old}[1]{ }
\newcommand{\commentout}[1]{ }

\newcommand{\K}{\mathbb{K}}

\newcommand{\uB}{\underline{\mathfrak{B}}}
\newcommand{\BB}{\mathfrak{B}}

\newcommand{\BC}{BC}

\DeclareMathOperator{\Int}{Int}
\DeclareMathOperator{\Th}{Th}

\newcommand{\abs}[1]{{\mid}#1{\mid}}

\newcommand{\UBB}{\underline{\BB}}

\let\subset\subseteq
\let\cref\Cref

\begin{document}
\title{The Birkhoff completion of finite lattices}
\author{Mohammad Abdulla\inst{1,2}, Johannes Hirth\inst{1,2}, Gerd Stumme\inst{1,2}} 
%
%\authorrunning{}
\institute{%
	Knowledge \& Data Engineering Group,
	University of Kassel, Germany
	\and
	Interdisciplinary Research Center for Information System Design,
	\mbox{University of Kassel, Germany}
}
	%\email{\{abdulla, hirth, stumme\}}@cs.uni-kassel.de}
\maketitle              % typeset the header of the contribution
\begin{abstract}
	We introduce the Birkhoff completion as the smallest distributive lattice in which a given finite lattice can be embedded as semi-lattice. We discuss its relationship to implicational theories, in particular to R. Wille's simply-implicational theories. By an example, we show how the Birkhoff completion can be used as a tool for ordinal data science. 
	% and truncated distributive lattices
	
%	 In this paper we investigate different approaches for fixing distributivity within lattices,
%	 focusing on  Birkhoff's construction and the strategic removal of elements. We discuss the implications of these methods on lattice structure and properties, comparing their advantages and limitations.Our findings contribute to the understanding of lattice structures and provide insights into addressing distributivity within lattice-based frameworks.

\end{abstract}

%%% Local Variables:
%%% mode: latex
%%% TeX-master: "2024"
%%% End:

\section{Introduction}
\label{sec:introduction}

Formal Concept Analysis (FCA) is a popular means for data analysis,
especially when at least one feature is not of nominal or interval
scale type. It is known since the beginning that the resulting concept
lattice can take the form of any finite\footnote{Throughout this paper
  we assume that our data --- the formal context --- are finite. If we
  would also allow for infinite contexts then the concept lattice can
  take the form of any complete lattice.} lattice. In real-world
applications, though, surprisingly often the concept lattice carries
more structure: being a distributive lattice which eventually is ``cut
off'' at the bottom.  \forlongversion{With \emph{truncated
    distributive lattices}, R, Wille has provided a formalisation of
  this observation \cite{wille2003truncated} which we will recall
  later in Section~\ref{sec-truncated-BC}.}

The likely reason for this unexpected emergence of distributivity is
related to the fact that a concept lattice is distributive if and only
if the implicational logic of its attributes can be described
completely by implications that have only one attribute in their
premise. The distributivity thus results from the prevalence of
implications with one-element premises. According to Wille (2003)
\cite{wille2003truncated} this is caused ``by the nature of human
thought namely that everyday thinking predominantly uses logical
inferences with one element premises.''

If a concept lattice is distributive, this comes with several
advantages. Besides a simpler --- and easier to communicate ---
implicational logic, one benefits from the fact that the concept
lattice can be decomposed into chains, which support an easier
interpretation of the data and a natural grid-like visualisation by an
embedding into the products of these chains. Furthermore the lattice
is isomorphic to the order ideals of the ordered set of its
meet-irreducible elements, which provides an interesting structure for
supporting the navigation in the concept lattice.

But what if the concept lattice of formal context is not (completely)
distributive? There are different possible explanations. The first is
of course that the data we are analysing have a more complex logical
structure. Another reason might be that the logic of our field of
interest has a simple implicational structure but that the formal
context which we have at hand is not perfectly representing it. It
might for instance miss some of the objects or attributes or contain
errors.

In this paper, we present a construction that may ``repair'' a
non-distributive concept lattice whenever we have reason to assume
that our context is not covering all potential objects of our domain
of interest. The underlying construction is not specific to FCA, and
all its building blocks --- in particular Birkhoff's representation
theorem for distributive lattices \cite{Birkhoff} --- are known in
lattice theory for a very long time. However, no-one has yet turned it
into a construction for embedding a non-distributive lattice into a
minimal distributive lattice.\footnote{For a reader with background in
  lattice theory or general algebra it will be immediately clear that
  such an embedding will only work if we do not require the embedding
  to preserve all lattice operations. We will discuss this in detail
  in Section~\ref{sec-fixing}.} Our proposed \emph{Birkhoff
  completion} is therefore not only of interest to FCA but can be seen
as a generic lattice-theoretical construction. We will therefore also
introduce the Birkhoff completion in purely lattice-theoretical terms
in Section~\ref{sec-fixing}. References to related work are
not collected in a separate section, but mentioned whenever
appropriate.  

All methods presented in this paper are implemented and provided by
the \texttt{conexp-clj} FCA research framework \cite{conexp-clj}.

%%% Local Variables:
%%% mode: latex
%%% TeX-master: "2024"
%%% End:

\section{Basics}
\label{sec:fca-basics}

In the following we recall some notions from order and lattice theory \cite{Birkhoff,Lattices}, and introduce some notations used in this work.

\iffalse A partially ordered set (or poset) $(P, \leq)$ consists of a nonempty set $P$ and a binary relation $\leq$ on $P$ such that $\leq$ satisfies the following properties:
\begin{enumerate}
	\item Reflexivity: $x \leq x$ for all $x \in P$;
	\item Antisymmetry: $x \leq y$ and $y \leq x$ implies $x = y$ for all $x, y \in P$;
	\item Transitivity: $x \leq y$ and $y \leq z$ implies $x \leq z$ for all $x, y, z \in P$.
\end{enumerate}

If $\leq$ satisfies the additional property:
\begin{enumerate}
	\setcounter{enumi}{3}
	\item Linearity: $x \leq y$ or $y \leq x$ for all $x, y \in P$,
\end{enumerate}
then $P$ is a total order or a chain. 
We write $x || y $ if $ x \nleq y \text{ and } y \nleq x$.
A linear extension of a partial order $\leq$ on a set $S$ is a total order $\leq_L$ on $S$ such that for all $x, y \in S$, if $x \leq y$ in the partial order $\leq$, then $x \leq_L y$ in the total order $\leq_L$. An element $a \in P$ is called an \emph{upper bound} for a subset $T \subseteq P$ if $a \geq x$ for all $x \in T$. If $a$ is the unique smallest upper bound of $T$, we say that $a$ is the \emph{supremum} of $T$.

Similarly, an element $b \in P$ is called a \emph{lower bound} for a subset $T \subseteq P$ if $b \leq x$ for all $x \in T$. If $b$ is the unique greatest lower bound of $T$, we say that $b$ is the \emph{infimum} of $T$.
Let \( P \) be an ordered set and let \( x, y \in P \),
we say \( x \) is covered by \( y \) , and write \( x \prec y \)  if
\( x < y \) and \( x \leq z < y \) implies \( z = x \).\fi

Let  \( P \) and \( Q \) be two partially ordered sets.
We say that \( P \) and \( Q \) are order-isomorphic, denoted by \( P \cong Q \), if there exists an injective map \( \varphi \) from \( P \) onto \( Q \) such that \( x \leq y \) in \( P \) if and only if \( \varphi(x) \leq \varphi(y) \) in \( Q \). 
Such a map is called an \textit{order-isomorphism}.
\iffalse Let $P_1, \ldots, P_n$ be ordered sets. The Cartesian product $P_1 \times \ldots \times P_n$ can be made into an ordered set by imposing the coordinatewise order defined by
\[
(x_1, \ldots, x_n) \leq (y_1, \ldots, y_n) \iff \forall i \, (x_i \leq y_i \text{ in } P_i).
\]
\fi
Associated with any ordered set is the corresponding familiy of down-sets, which plays a major role in our paper.
Let $P$ be an ordered set and $Q \subseteq P$.
 We say that $Q$ is a \emph{down-set} (or \emph{$($order$)$ ideal}) if, whenever $x \in Q$, $y \in P$ and $y \leq x$, we have $y \in Q$. Dually, $Q$ is an \emph{up-set} (or \emph{$($order$)$ filter}) if, whenever $x \in Q$, $y \in P$ and $x \leq y$, we have $y \in Q$. The down-set of $x\in L$ is $\downarrow x :=\{a\in L \mid a\leq x\}$, and the  up-set of $x\in L$ is $\uparrow x :=\{a\in L \mid a\geq x\}$.
  We denote the family of the order ideals of $P$ by $\mathcal{I}(P)$, and the family of its order filters by $\mathcal{F}(P)$. They are both lattices, under the inclusion order.

A partially ordered set $L$ is called a \emph{lattice} if, for any two elements $x,y \in L$, their infimum $x\wedge y$ and supremum $x\vee y$ exist in $L$.
\iffalse $L$ is called a \emph{complete lattice} if  all subsets $X\subseteq L$ have an infimum $\bigwedge X$ and a supremum $\bigvee X$ in $L$.
\fi In a lattice, infimum and supremum are also called \emph{meet} and \emph{join}, resp.
A join-semilattice  is a partially ordered set that has a join for any nonempty finite subset. Dually, a meet-semilattice  is a partially ordered set which has a meet for any nonempty finite subset. An element $x$ of a (finite) lattice $L$ is called \emph{meet-irreducible} if $\bigvee X = x$ implies $x\in X$ for all $X\subseteq L$. Dually, $x$ is called \emph{join-irreducible} if $\bigwedge X = x$ implies $x\in X$ for all $X\subseteq L$. The set of all meet-irreducible elements of $L$ is denoted by $\mathcal{M}(L)$, and the set of all its join-irreducibles is denoted by $\mathcal{J}(L)$.
\iffalse A bounded lattice is a lattice $L$ that additionally has a greatest element (also called maximum, or top element, and denoted by $1$ ) and a least element (also called minimum, or bottom, denoted by $0$ ), which satisfy:
\[
0 \leq x \leq 1 \quad \text{for every } x \in L.
\]
 If \( L \) is a lattice with least element \( 0 \),then \( a \in L \) is called an \emph{atom} if \( 0 \prec a \). 

If \( (L, \vee, \wedge) \) is a bounded lattice, then we say that $y \in L$ is a complement of $x \in L$ if $x \land y = 0$ and $x \lor y = 1$. In this case, we say that $x$ is a complemented element of $L$. Clearly, every complement of a complemented element is itself complemented.
\fi We call $S \neq \emptyset$ a \emph{sublattice} of $L$ if $(a,b\in S\Rightarrow a\vee b\in S \text{ and } a\wedge b\in S)$ holds.
\iffalse Let \( L = (L, \vee, \wedge) \) be a lattice and $a, b \in L$. Then the interval $[a, b]$  is defined as follows:
\[ [a, b] = \{x \in L : a \leq x \leq b\} \]
An element \( z \) of \( [a, b] \) is called a relative complement of \( x \) in \( [a, b] \) if \( x \vee z = b \) and \( x \wedge z = a \), i.e., \( z \) is a complement of \( x \) in the sublattice \( ([a, b], \vee, \wedge) \) of \( L \).

 Let $L$ and $K$ be lattices. 
Define $\lor$ and $\land$ coordinate-wise on $L \times K$, as follows:

\[
(\ell_1, k_1) \lor (\ell_2, k_2) = (\ell_1 \lor \ell_2, k_1 \lor k_2)
\]

\[
(\ell_1, k_1) \land (\ell_2, k_2) = (\ell_1 \land \ell_2, k_1 \land k_2)
\]
\fi

Let $L$ and $K$ be lattices. A map $f : L \rightarrow K$ is said to be a lattice homomorphism iff for all $a, b \in L$,
$f(a \lor b) = f(a) \lor f(b)$ and $f(a \land b) = f(a) \land f(b)$. A bijective lattice homomorphism is a lattice isomorphism. If $f : L \rightarrow K$ is a one-to-one homomorphism, then the sublattice $f(L)$ of $K$ is isomorphic to $L$ and we refer to $f$ as an embedding (of $L$ into $K$).

A map $\phi : L \rightarrow K$ between two $($finite$)$ lattices $L$ and $K$ is called join-preserving if $\phi(\bigvee X) = \bigvee \phi(X)$ holds for all $X \subseteq L$. Dually,$\phi$ is  meet preserving if   $\phi(\bigwedge X) = \bigwedge \phi(X)$ holds for all $X \subseteq L$.
\iffalse While distributive lattices serve as our primary tool for understanding structural relationships within partially ordered sets the concept of modularity introduces a fundamental aspect that enriches our understanding further.
 \begin{definition}\label{def:modular}
	A lattice $L$ is said to be modular if for all $x, y, z \in L$:
	\[
	x \leq z \Rightarrow x \lor (y \land z) = (x \lor y) \land z
	\]
	
\end{definition}
 In nonmodular lattices, we shall be interested in modular pairs and dual modular pairs, which were introduced by Wilcox [1938], [1939]. We say that an ordered pair \( (a, b) \) of elements of a lattice \( L \) is a modular pair if for all \( c \in L \), 
\[ c \leq b \implies c \vee (a \wedge b) = (c \vee a) \wedge b. \]

We say that \( (a, b) \) is a dual modular pair if for all \( c \in L \), 
\[ c \geq b \implies c \wedge (a \vee b) = (c \wedge a) \vee b. \]

\begin{figure}[h]
	% \begin{center}
		\centering
		\begin{minipage}{0.3\textwidth} 
			\centering 
\input{figs/N5}
			
		\end{minipage} 
		%\end{center}
		\caption{The nondistributive lattice $N_5$}
		% \begin{center}
			\centering
			
			%\end{center}	
		\end{figure}
		\fi
		Distributivity in a lattice governs how the meet and join operations interact, ensuring consistency in lattice transformations.
				\begin{definition}\label{def:distributive}
			A lattice $L$ is said to be distributive if for all $x, y, z \in L$:
			\[
			x \land (y \lor z) = (x \land y) \lor (x \land z)
			\]
			
		\end{definition}
		\iffalse 
		\begin{figure}[h]
			% \begin{center}
				\centering
				\begin{minipage}{0.3\textwidth}
					\centering
					\input{figs/M3}
				\end{minipage}
				%\end{center}
				\caption{The nondistributive  lattice $M_3$}
				% \begin{center}
					\centering
					
					%\end{center}	
				
				\end{figure}
				\fi 
				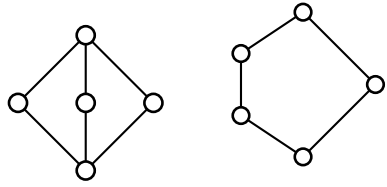
\begin{figure}[t]
					% \begin{center}
						\centering
						\begin{minipage}{0.3\textwidth}
							\centering
												{\unitlength 0.6mm
	\begin{picture}(60,40)%
		\put(0,0){%
			\begin{diagram}{60}{60}
				\Node{1}{30}{0}
				\Node{2}{15}{15}
				\Node{3}{30}{15}
				\Node{4}{45}{15}
				\Node{5}{30}{30}
				\Edge{1}{2}
				\Edge{1}{3}
				\Edge{1}{4}
				\Edge{2}{5}
				\Edge{3}{5}
				\Edge{4}{5}
				%\Numbers\Edge{2}{3}
				\leftObjbox{2}{3}{1}{}
				\rightObjbox{3}{3}{1}{}
				\leftAttbox{2}{3}{1}{}
				\rightAttbox{3}{3}{1}{}
				\rightAttbox{4}{3}{1}{}
		\end{diagram}}
\end{picture}}			
						\end{minipage}
						\begin{minipage}{0.3\textwidth}
							\centering
									{\unitlength 0.55mm
	\begin{picture}(60,40)%
		\put(0,0){%
			\begin{diagram}{60}{60}
%				\Node{1}{15}{0}
				\Node{1}{15}{5}
				\Node{2}{0}{15}
%				\Node{3}{30}{15}
				\Node{3}{32.5}{22.5}
				\Node{4}{0}{30}            
				\Node{5}{15}{40}     
				\Edge{1}{2}
				\Edge{1}{3}
				\Edge{2}{4}
				\Edge{4}{5}
				\Edge{3}{5}
				
				\leftObjbox{2}{3}{1}{}
				\leftObjbox{4}{3}{1}{}
				\rightObjbox{3}{3}{1}{}
				\rightObjbox{5}{3}{1}{}
				\leftAttbox{2}{3}{1}{}
				\leftAttbox{4}{3}{1}{}
				\rightAttbox{3}{3}{1}{}
				\rightAttbox{5}{3}{1}{}
				\rightAttbox{1}{3}{1}{}
				
		\end{diagram}}
\end{picture}}	
						\end{minipage}
						\caption{ The nondistributive  lattices $M_3$ and $N_5$}
				\end{figure}
				\begin{theorem}
					A lattice $L$ is distributive if and only if there exists no sublattice $A \subseteq L$ isomorphic to either $M_3$ or $N_5$.
				\end{theorem}
\section{Fixing non distributivity based on Birkhoff's Representation Theorem}
 \label{sec-fixing}

% One approach to address the lack of distributivity in a lattice $L$ is to consider its join-irreducible elements.
Our approach for fixing non-distributivity is based on Birkhoff's Representation Theorem \cite{Birkhoff}, which states that for any finite distributive lattice $L$, the lattice $\mathcal{I}(\mathcal{J}(L))$ of down-sets of the join-irreducible elements of $L$ is isomorphic to $L$. It's important to note that we focus on this variation of the theorem tailored for finite lattices, as our work primarily involves finite structures.

\begin{theorem}[Birkhoff's Representation Theorem]
	\label{theorem-birkhoff}

	If $L$ is a finite\footnote{This theorem holds --- like many of the statements in this paper --- also for specific infinite lattices, the so-called doubly founded lattices. As our focus is on applications in data analysis, we consider only the finite case in the sequel, which is easier to handle.} distributive lattice, then the map $\eta \colon L \rightarrow \mathcal{I}(J(L))$ defined by $\eta(a) = \{ x \in \mathcal{J}(L) \mid x \leq a \}$ is an isomorphism of $L$ onto $(\mathcal{I}(J(L)),\subseteq)$.
		Conversely, for every finite ordered set, the closure system of all order ideals is a  distributive lattice \( D \) in which
\[ \mathcal  J(D) = \{ \downarrow x \,|\, x \in P \} \]
is supremum-dense.

\end{theorem}
We refer to Ganter and Wille (1999) \cite{fca-book} for the proof.

Although this theorem is usually based on the down-sets of the join-irreducible elements of $L$, the same also holds for the up-sets of the meet-irreducible elements of $L$: 
We have $ (\mathcal{I}(\mathcal{J}(L)), \subseteq) \cong (\mathcal{F}(\mathcal{M}(L)), \supseteq) $ for distributive lattices
because in that case $(\mathcal{J}(L),\subseteq)\cong(\mathcal{M}(L),\subseteq)$ holds (see Exercise 5.7  in Davey and Priestley, 2002 \cite{Lattices}) and because the complement of an order filter of an ordered set is always an order ideal and vice versa. 
%	For any finite distributive lattice \( L \), by Corollary 1,   Consequently, we conclude that $ (J(L),\leq) \cong ( M(L) ,\leq)$ if $L$ is  a finite distributive lattice. 
As we want to link our work with observations about the implications between attributes in a formal context in Section~\ref{sec-birkhoff-context}, we prefer to continue with the equivalent dual representation: 

\begin{theorem}[Dual version of Birkhoff's Representation Theorem]
	\label{theorem-dual-birkhoff}
	If $L$ is a finite distributive lattice, then the map $\eta \colon L \rightarrow \mathcal{F}(\mathcal M(L))$ defined by $\eta(a) = \{ x \in \mathcal{M}(L) \mid x \geq a \}$ is an isomorphism of $L$ onto $(\mathcal{F}(\mathcal M(L)),\supseteq).$
	Conversely, for every finite ordered set $P$, the closure system of all order filters is a  distributive lattice in which
	\[ \mathcal M(D) = \{ \uparrow x \,|\, x \in P \} \]
	is infimum-dense.
\end{theorem}

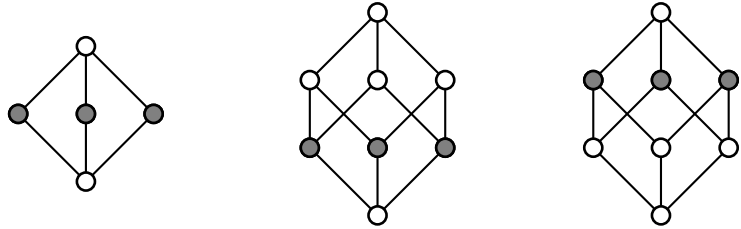
\begin{figure}[h]
	% \begin{center}
		\centering
		\begin{minipage}{0.3\textwidth}
			\centering
						{\unitlength 0.6mm
	\begin{picture}(60,60)%
		\put(0,0){%
			\begin{diagram}{60}{60}
				\Node{1}{30}{0}
				\Node{2}{15}{15}
				\Node{3}{30}{15}
				\Node{4}{45}{15}
				\Node{5}{30}{30}
				\Edge{1}{2}
				\Edge{1}{3}
				\Edge{1}{4}
				\Edge{2}{5}
				\Edge{3}{5}
				\Edge{4}{5}
				%\Numbers\Edge{2}{3}
				\leftObjbox{2}{3}{1}{}
				\rightObjbox{3}{3}{1}{}
				\leftAttbox{2}{3}{1}{}
				\rightAttbox{3}{3}{1}{}
				\rightAttbox{4}{3}{1}{}

		\end{diagram}}
		\put(15,15){\ColorNode{gray}}
		\put(30,15){\ColorNode{gray}}
		\put(45,15){\ColorNode{gray}}
\end{picture}}			
		\end{minipage}
		\begin{minipage}{0.3\textwidth}
			\centering
		\end{minipage}
		\begin{minipage}{0.3\textwidth}
			\centering
						{\unitlength 0.6mm
	\begin{picture}(60,75)%
		\put(0,0){%
			\begin{diagram}{60}{75}
				\Node{1}{30}{0}
				\Node{2}{15}{15}
				\Node{3}{30}{15}
				\Node{4}{45}{15}
				\Node{5}{15}{30}
				\Node{6}{30}{30}
				\Node{7}{45}{30}
				\Node{8}{30}{45}
				\Edge{1}{2}
				\Edge{1}{3}
				\Edge{1}{4}
				\Edge{2}{5}
				\Edge{2}{6}
				\Edge{3}{5}
				\Edge{3}{7}
				\Edge{4}{6}
				\Edge{4}{7}
				\Edge{5}{8}
				\Edge{6}{8}
				\Edge{7}{8}
				%\Numbers\Edge{2}{3}
				\leftObjbox{2}{3}{1}{}
				\rightObjbox{3}{3}{1}{}
				\rightObjbox{4}{3}{1}{}
				\leftAttbox{5}{3}{1}{}
				\rightAttbox{6}{3}{1}{}
				\rightAttbox{7}{3}{1}{}
		\end{diagram}}
		\put(15,15){\ColorNode{gray}}
		\put(30,15){\ColorNode{gray}}
		\put(45,15){\ColorNode{gray}}
\end{picture}}			
		\end{minipage}
		\begin{minipage}{0.3\textwidth}
			\centering
						{\unitlength 0.6mm
	\begin{picture}(60,75)%
		\put(0,0){%
			\begin{diagram}{60}{75}
				\Node{1}{30}{0}
				\Node{2}{15}{15}
				\Node{3}{30}{15}
				\Node{4}{45}{15}
				\Node{5}{15}{30}
				\Node{6}{30}{30}
				\Node{7}{45}{30}
				\Node{8}{30}{45}
				\Edge{1}{2}
				\Edge{1}{3}
				\Edge{1}{4}
				\Edge{2}{5}
				\Edge{2}{6}
				\Edge{3}{5}
				\Edge{3}{7}
				\Edge{4}{6}
				\Edge{4}{7}
				\Edge{5}{8}
				\Edge{6}{8}
				\Edge{7}{8}
				%\Numbers\Edge{2}{3}
				\leftObjbox{2}{3}{1}{}
				\rightObjbox{3}{3}{1}{}
				\rightObjbox{4}{3}{1}{}
				\leftAttbox{5}{3}{1}{}
				\rightAttbox{6}{3}{1}{}
				\rightAttbox{7}{3}{1}{}
		\end{diagram}}
		\put(15,30){\ColorNode{gray}}
		\put(30,30){\ColorNode{gray}}
		\put(45,30){\ColorNode{gray}}
\end{picture}}			
		\end{minipage}
		%\end{center}
		\caption{ $M_3$ (left) , its down-set $BC$ (middle) and its up-set $BC$(Right).}
		
		  \label{fig-M3}
	\end{figure}

	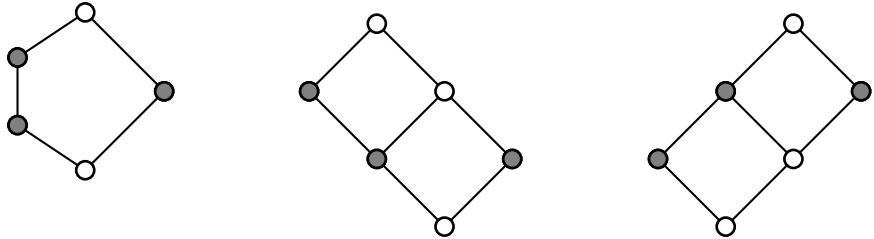
\begin{figure}[h]
		% \begin{center}
			\centering
			\begin{minipage}{0.3\textwidth}
				\centering
								{\unitlength 0.6mm
	\begin{picture}(60,60)%
		\put(0,0){%
			\begin{diagram}{60}{60}
%				\Node{1}{15}{0}
				\Node{1}{15}{5}
				\Node{2}{0}{15}
%				\Node{3}{30}{15}
				\Node{3}{32.5}{22.5}
				\Node{4}{0}{30}            
				\Node{5}{15}{40}     
				\Edge{1}{2}
				\Edge{1}{3}
				\Edge{2}{4}
				\Edge{4}{5}
				\Edge{3}{5}
				
				\leftObjbox{2}{3}{1}{}
				\leftObjbox{4}{3}{1}{}
				\rightObjbox{3}{3}{1}{}
				\rightObjbox{5}{3}{1}{}
				\leftAttbox{2}{3}{1}{}
				\leftAttbox{4}{3}{1}{}
				\rightAttbox{3}{3}{1}{}
				\rightAttbox{5}{3}{1}{}
				\rightAttbox{1}{3}{1}{}
		\end{diagram}}
		\put(0,15){\ColorNode{gray}}
%		\put(30,15){\ColorNode{gray}}
		\put(32.5,22.5){\ColorNode{gray}}
		\put(0,30){\ColorNode{gray}}
\end{picture}}			
			\end{minipage}
			\begin{minipage}{0.3\textwidth}
				\centering
			\end{minipage}
			\begin{minipage}{0.3\textwidth}
				\centering
								{\unitlength 0.6mm
	\begin{picture}(60,75)%
		\put(0,0){%
			\begin{diagram}{60}{75}
				\Node{1}{30}{0}
				\Node{2}{15}{15}
				\Node{3}{45}{15}
				\Node{4}{30}{30}
				\Node{5}{0}{30}
				\Node{6}{15}{45}

				\Edge{1}{2}
				\Edge{1}{3}
				\Edge{3}{4}
				\Edge{2}{4}
				\Edge{6}{5}
				\Edge{6}{4}
				\Edge{5}{2}

				%\Numbers\Edge{2}{3}
				
		\end{diagram}}
		\put(15,15){\ColorNode{gray}}
		\put(0,30){\ColorNode{gray}}
		\put(45,15){\ColorNode{gray}}

\end{picture}}			
			\end{minipage}
			\begin{minipage}{0.3\textwidth}
								{\unitlength 0.6mm
	\begin{picture}(60,75)%
		\put(0,0){%
			\begin{diagram}{60}{75}
				\Node{1}{30}{0}
				\Node{2}{15}{15}
				\Node{3}{45}{15}
				\Node{4}{30}{30}
				\Node{5}{45}{45}
				\Node{6}{60}{30}
				
				\Edge{1}{2}
				\Edge{1}{3}
				\Edge{3}{4}
				\Edge{2}{4}
				\Edge{5}{4}
				\Edge{5}{6}
				\Edge{6}{3}

				%\Numbers\Edge{2}{3}
				
		\end{diagram}}
		\put(15,15){\ColorNode{gray}}
		\put(30,30){\ColorNode{gray}}
		\put(60,30){\ColorNode{gray}}
\end{picture}}			

			\end{minipage}
			%\end{center}
			\caption{ $N_5$ (left), its down-set $BC$ (middle) and its up-set $BC$(Right).}
			
		\end{figure}

\begin{corollary}
	Let $L$ be a finite lattice. Then the following statements are equivalent:
	\begin{enumerate}[label=(\roman*)]
		\item $L$ is distributive.
		\item $L \cong (\mathcal{I}(\mathcal{J}(L)),\subseteq)$.
		\item $L \cong (\mathcal{F}(\mathcal{M}(L)),\supseteq)$.
		\item $L$ is isomorphic to the lattice of all down-sets of some finite ordered set.
		\item $L$ is isomorphic to the lattice of all up-sets of some finite ordered set.
	\end{enumerate}
\end{corollary}

This implies that starting from any finite lattice $L$, one obtains two distributive lattices $\mathcal{I}(\mathcal{J}(L))$ and $\mathcal{F}(\mathcal{M}(L))$, which are isomorphic to $L$ if and only if $L$ is distributive. This gives rise to the following constructions. 

%Moreover, since $O(J(L))$ captures the essential distributive structure of $L$, many properties and concepts studied in the context of distributive lattices can be explored within $O(J(L))$.\doto{What do you mean by this?}

\begin{definition}
Let $(L,\leq)$ be a finite lattice. 
The \emph{down-set Birkhoff completion} of $L$ is the lattice $(\mathcal{I}(\mathcal J(L)),\subseteq)$. The \emph{up-set Birkhoff completion} of $L$ is the lattice $(\mathcal{F}(\mathcal M(L)),\supseteq)$. When we simply talk about the \emph{Birkhoff completion}, we refer to the latter, and denote it by $BC(L)$.
\end{definition}

\begin{remark}
	\label{remark-BC}
	In terms of FCA, these two constructions are isomorphic to the two concept lattices $\uB(\mathcal{J}(L), \mathcal{J}(L),\not\geq)$ and $\uB(\mathcal{M}(L), \mathcal{M}(L),\not\geq)$. In fact, the order ideals of $\mathcal{J}(L)$ are exactly the concept extents of $\uB(\mathcal{J}(L), \mathcal{J}(L),\not\geq)$, and the order filters of  $\mathcal{F}(L)$ are exactly the concept extents of $\uB(\mathcal{M}(L), \mathcal{M}(L),\not\geq)$ \cite{wille1985tensorial}. By abuse of notation, we will therefore also refer to these concept lattices as down-set Birkhoff completion and (up-set) Birkhoff completion, resp. 
\end{remark}

%Similarly, the Birkhoff completion on the dual lattice of \( L \), denoted by \( BC(L^\partial) \), can be defined as follows:

%\begin{definition}
%	Let \( L \) be a finite lattice. 
%	The Birkhoff completion on the dual lattice of \( L \), denoted by \( BC(L^\partial) \), is the lattice obtained by applying the Birkhoff construction \( \mathfrak{B}(M(L), M(L), \not\leq) \) to the meet-irreducible elements of \( L \).
%	\end{definition} 

The Birkhoff completion $BC(L)$ can be understood as a way to extend the original lattice $L$ in a somewhat minimal way to a distributive lattice while preserving its essential structure.\footnote{This is in a certain sense complementary to the construction of the free distributive completion of a partial complete lattice, which was introduced in Stumme (1998) \cite{stumme98free}.}

\begin{theorem}
  \label{theorem-embedding}
  Let $L$ be a finite lattice. 
  
  \begin{enumerate}
  \item $BC(L)$ is distributive.
  \item The map \[\iota	\colon L \rightarrow BC(L)\ :
      x \mapsto \enspace \uparrow x \cap \mathcal M(L)\]is a join-semilattice embedding.
    % Dually, the map  \[ \varphi_2: L \rightarrow BC(L^\partial)  \] \[ x \mapsto (L \setminus \downarrow_ x, \downarrow_x) \] is a join-preserving embedding.
  \item For any distributive lattice $\hat{L}$ for which exists a join-semilattice embedding $\varphi \colon L \rightarrow \hat{L}$, there exists a  join-semilattice embedding $\varepsilon\colon BC(L) \rightarrow \hat{L}$ such that $\varphi = \varepsilon \circ \iota$ holds.
  \end{enumerate}
\end{theorem}
% Remark: I renamed \varphi by \iota, and \varpsi by \varphi.
\begin{proof}
  \begin{enumerate}
  \item This follows directly from Theorem~\ref{theorem-dual-birkhoff}, since $\BC(L) = (\mathcal{F}(\mathcal{M}(L)),\supseteq) $.
  \item First, we note that for $a, b \in L$ the equality
    $\uparrow (a \vee b) = \uparrow a \:\cap \uparrow b$ holds. This
    implies that $\iota (a \vee b) = \iota (a) \cap \iota (b)$. Thus,
    $\iota$ is joins preserving. In a finite lattice every element
    $x\in L$ can be uniquely identified by the set of all
    meet-irreducible elements $m\in \mathcal M(L)$ for which $m \geq x$, i.e.,
    $x=\bigwedge \{m\in \mathcal M(L)\mid m\geq x\}$. Thus, is for $x,y\in L$
    with $x\neq y$ the image $\iota(x)\neq \iota(y)$. Concluding,
    $\iota$ is injective and therefore a join-semilattice embedding.
  \item Let $\kappa: \BC(L)\to L$ with
    $\kappa(A) \coloneqq \{ x \in L\mid \iota(x) \subseteq A \}$ and
    $\varepsilon: \BC(L)\to \hat{L}$ with
    $\varepsilon(A) \coloneqq \bigwedge_{\hat{L}} \{ \varphi(x) \mid x
    \in \kappa(A) \}$ be two maps. For $x,y\in L$, we find that
    $\iota(y)\subseteq \iota(x) \iff y \geq x$. Combined with
    the fact that $\varphi$ is an order-embedding we follow that
    $(\varepsilon \circ \iota) (x) = \varphi(x)$.

    Next, we want to show that $\varepsilon$ is join-preserving, i.e.,
    $ \varepsilon(A \vee_{\BC(L)} B) = \varepsilon(A)
    \vee_{\hat{L}}\varepsilon(B)$. For this, we first transform both
    sides of the equation to then show their equivalence:
    \begin{align*}
      \varepsilon(A) \vee_{\hat{L}} \varepsilon(B) &= \bigwedge_{\hat{L}} \{ \varphi(x) \mid x \in \kappa(A) \} \vee_{\hat{L}} \bigwedge_{\hat{L}} \{ \varphi(y) \mid y \in \kappa(B) \}& \\
      \text{(distributivity of $\hat{L})$} \quad         &= \bigwedge_{\hat{L}} \{ \varphi(x) \vee_{\hat{L}} \varphi(y) \mid x \in \kappa(A), y \in \kappa(B) \} &  \\
      \text{($\varphi$ is join preserving)} \quad &= \bigwedge_{\hat{L}} \{ \varphi(x \vee y) \mid x \in \kappa(A), y \in \kappa(B) \} &  \\
                                                   &= \bigwedge_{\hat{L}} \{ \varphi(z) \mid z = x \vee y, x \in \kappa(A), y \in \kappa(B) \},&\\
       \varepsilon(A \vee_{\BC(L)} B) &= \bigwedge_{\hat{L}} \{ \varphi(x) \mid x \in \kappa(A\vee_{\BC(L)} B) \}&\\
                               &=\bigwedge_{\hat{L}} \{ \varphi(z) \mid z \in \kappa(A \cap B) \}& \\
      &= \bigwedge_{\hat{L}} \{ \varphi(z) \mid z \in \kappa(A) \cap \kappa(B) \}.& \\
    \end{align*}
    We proof their equality by showing for a $z\in L$ that
    $z\in \kappa(A) \cap \kappa(B)$ iff there exist
    $x\in \kappa(A),y\in \kappa(B)$ with
    $z=x\vee_{L}y$. [$\Rightarrow$] This direction holds for $x=z$ and
    $y=z$. [$\Leftarrow$] We can follow from $z= x\vee_{L}y$ that
    $z\geq x,y$. From the fact that $x\in \kappa(A)$,
    $y\in \kappa(B)$, we can follow that
    $\iota(z)\subseteq \iota(x) \subseteq A$ and that
    $\iota(z)\subseteq \iota(y) \subseteq B$. Thus, $z\in \kappa(A)$
    and $z\in \kappa(B)$. Concluding,
    $\varepsilon(A) \vee_{\hat{L}} \varepsilon(B) = \varepsilon(A
    \vee_{\BC(L)} B)$ and $\varepsilon$ is join-preserving.

    It remains to be shown that $\varepsilon$ is injective.  Let
    $A,B\in \BC(L)$ with $A\neq B$. From this we can follow WLOG that
    there exists an $x \in A$ with $x\not\in B$ and that there is no
    $y \in B$ with $y\leq x$ (otherwise $x$ would have been in
    $B$). This implies that $x\not\in \kappa(B)$.  From the fact that
    $\varphi$ is injective, we can follow that there is no
    $y\in \kappa(B)$ with $\varphi(x)=\varphi(y)$. Moreover, since
    $\varphi$ is an order embedding, we can follow that there is no
    $y \in B$ such $\varphi(y)\leq \varphi(x)$. Hence,
    $\varepsilon(B) \wedge_{\hat{L}} \varphi(x) \neq \varepsilon(B)$,
    but $\varepsilon(A) \wedge_{\hat{L}} \varphi(x) =
    \varepsilon(A)$. Concluding, $\varepsilon$ is injective and
    therefore a join-semilattice embedding.\qed
  \end{enumerate}
\end{proof}

Of course we cannot expect that we can always find a lattice embedding of $L$ into $BC(L)$ --- because then $BC(L)$ would have a non-distributive sub-lattice. Note that if we require the embedding to be an order embedding only, the embedding could be even smaller, as \cref{fig:smallest-order-embedding} shows. 

\begin{figure}[t]
  \centering
  {\unitlength 0.6mm
  \begin{picture}(60,120)%
    \put(0,0){%
      \begin{diagram}{60}{120}
        \Node{1}{30}{60}
        \Node{2}{15}{75}
        \Node{3}{45}{75}
        \Node{4}{15}{90}
        \Node{5}{30}{90}
        \Node{6}{45}{90}
        \Node{7}{30}{105}
        \Node{8}{30}{45}
        \Node{9}{15}{30}
        \Node{10}{45}{30}
        \Node{11}{15}{15}
        \Node{12}{30}{15}
        \Node{13}{45}{15}
        \Node{14}{30}{0}

        \Edge{1}{2}
        \Edge{1}{3}
        \Edge{2}{5}
        \Edge{3}{5}
        \Edge{4}{7}
        \Edge{4}{2}
        \Edge{3}{6}
        \Edge{5}{7}
        \Edge{6}{7}
        \Edge{1}{8}        
        \Edge{8}{9}
        \Edge{8}{10}
        \Edge{9}{11}
        \Edge{9}{12}
        \Edge{10}{12}
        \Edge{10}{13}
        \Edge{11}{14}
        \Edge{12}{14}
        \Edge{13}{14}
        
      \end{diagram}}
  \end{picture}
  \hspace{1cm}
  \begin{picture}(60,120)%
    \put(0,0){%
      \begin{diagram}{60}{120}
        \Node{1}{30}{60}
        \Node{2}{15}{75}
        \Node{3}{45}{75}
        \Node{4}{15}{90}
        \Node{5}{30}{90}
        \Node{6}{45}{90}
        \Node{7}{30}{105}
        \Node{8}{30}{45}
        \Node{9}{15}{30}
        \Node{10}{45}{30}
        \Node{11}{15}{15}
        \Node{12}{30}{15}
        \Node{13}{45}{15}
        \Node{14}{30}{0}

        \Node{15}{30}{30}
        \Node{16}{30}{75}

        \Edge{1}{2}
        \Edge{1}{3}
        \Edge{2}{5}
        \Edge{3}{5}
        \Edge{4}{7}
        \Edge{4}{2}
        \Edge{3}{6}
        \Edge{5}{7}
        \Edge{6}{7}
        \Edge{1}{8}        
        \Edge{8}{9}
        \Edge{8}{10}
        \Edge{9}{11}
        \Edge{9}{12}
        \Edge{10}{12}
        \Edge{10}{13}
        \Edge{11}{14}
        \Edge{12}{14}
        \Edge{13}{14}

        \Edge{11}{15}
        \Edge{13}{15}
        \Edge{8}{15}

        \Edge{1}{16}        
        \Edge{4}{16}        
        \Edge{6}{16}        
        
      \end{diagram}}
  \end{picture}
}	
%%% Local Variables:
%%% mode: latex
%%% TeX-master: "../2024"
%%% End:
  \caption{A lattice for which there exists an order embedding into a distributive lattice that has fewer elements than the up-set or down-set Birkhoff completion. }
  \label{fig:smallest-order-embedding}
\end{figure}
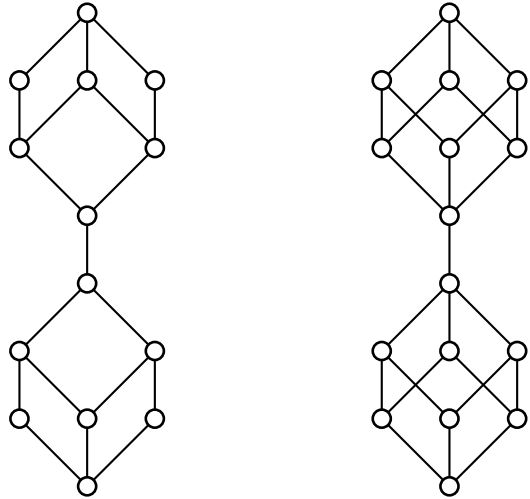

\begin{figure}[t]
  % \begin{center}
  \centering
		\begin{minipage}{0.3\textwidth}
			\centering
										{\unitlength 0.6mm
								\begin{picture}(60,60)%
									\put(0,0){%
										\begin{diagram}{60}{60}
											\Node{1}{30}{0}
											\Node{2}{15}{15}
											\Node{3}{45}{15}
											\Node{4}{15}{30}
											\Node{5}{30}{45}
											\Node{6}{45}{30}
											\Node{7}{30}{15}
											\Edge{1}{2}
											\Edge{1}{3}
											\Edge{4}{7}
											\Edge{4}{2}
											\Edge{3}{6}
											\Edge{6}{7}
											\Edge{1}{7}
											\Edge{5}{6}
											\Edge{5}{4}

			\leftObjbox{2}{3}{1}{J}
			\rightObjbox{3}{3}{1}{J}
			\rightAttbox{3}{3}{1}{M}
			\leftAttbox{4}{3}{1}{M}
			\leftAttbox{2}{3}{1}{M}	
			\rightObjbox{7}{3}{1}{J}
			\rightAttbox{6}{3}{1}{M}

											%\Numbers\Edge{2}{3}
											
									\end{diagram}}
									
							\end{picture}}			
						
		\end{minipage}
		\begin{minipage}{0.3\textwidth}
			\centering
						{\unitlength 0.6mm
	\begin{picture}(60,75)%
		\put(0,0){%
			\begin{diagram}{60}{75}
				\Node{1}{30}{0}
				\Node{2}{15}{15}
				\Node{3}{30}{15}
				\Node{4}{45}{15}
				\Node{5}{15}{30}
				\Node{6}{30}{30}
				\Node{7}{45}{30}
				\Node{8}{30}{45}
				\Edge{1}{2}
				\Edge{1}{3}
				\Edge{1}{4}
				\Edge{2}{5}
				\Edge{2}{6}
				\Edge{3}{5}
				\Edge{3}{7}
				\Edge{4}{6}
				\Edge{4}{7}
				\Edge{5}{8}
				\Edge{6}{8}
				\Edge{7}{8}
				%\Numbers\Edge{2}{3}
				\leftObjbox{2}{3}{1}{J}
				\rightObjbox{3}{3}{1}{J}
				\rightObjbox{4}{3}{1}{J}

		\end{diagram}}
		
\end{picture}}			
		\end{minipage}
		\begin{minipage}{0.3\textwidth}
			\centering
							{\unitlength 0.6mm
					\begin{picture}(60,75)%
										\put(0,0){%
											\begin{diagram}{60}{75}
												\Node{1}{30}{0}
												\Node{2}{15}{15}
												\Node{3}{45}{15}
												\Node{4}{30}{30}
												\Node{5}{45}{45}
												\Node{6}{60}{30}
												\Node{7}{30}{60}
												\Node{8}{15}{45}
												\Node{9}{0}{30}

												\Edge{1}{2}
												\Edge{1}{3}
												\Edge{3}{4}
												\Edge{2}{4}
												\Edge{5}{4}
												\Edge{5}{6}
												\Edge{6}{3}
												\Edge{9}{8}
												\Edge{7}{8}
												\Edge{6}{7}				  
												\Edge{9}{2}
												\Edge{3}{8}

			\leftAttbox{8}{3}{1}{M}
			\rightAttbox{5}{3}{1}{M}
			\rightAttbox{6}{3}{1}{M}	
			\leftAttbox{9}{3}{1}{M}	

										\end{diagram}}

								\end{picture}}		
		\end{minipage}
		\begin{minipage}{0.3\textwidth}
			\centering
		\end{minipage}
		\caption{ $L$ (left), its down-set $BC$ (middle) and its up-set $BC$ (right).}
		\label{fig-S7-completions}
\end{figure}
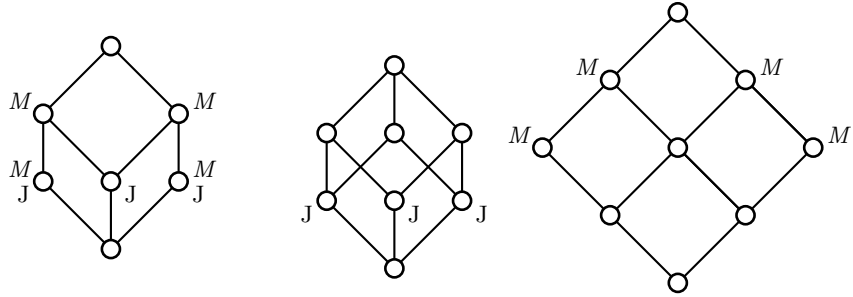
		
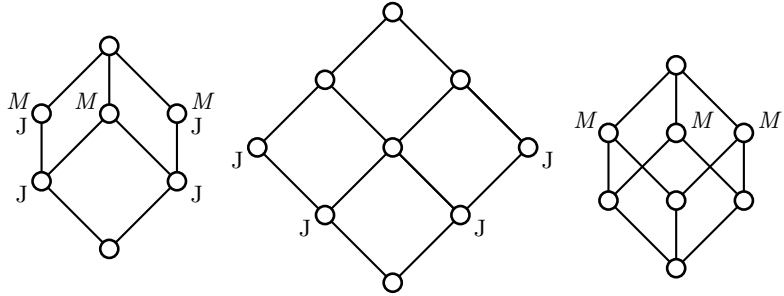
\begin{figure}
		\begin{minipage}{0.3\textwidth}
			\centering
										{\unitlength 0.6mm
								\begin{picture}(60,60)%
									\put(0,0){%
										\begin{diagram}{60}{60}
											\Node{1}{30}{0}
										\Node{2}{15}{15}
										\Node{3}{45}{15}
										\Node{4}{15}{30}
										\Node{5}{30}{30}
										\Node{6}{45}{30}
										\Node{7}{30}{45}
										\Edge{1}{2}
										\Edge{4}{2}
										\Edge{3}{1}
										\Edge{3}{5}
										\Edge{3}{6}
										\Edge{4}{7}
										\Edge{5}{7}
										\Edge{6}{7}
									    \Edge{2}{5}
										\leftObjbox{2}{3}{1}{J}
										\rightObjbox{3}{3}{1}{J}
										\leftAttbox{4}{3}{1}{M}
										\leftObjbox{4}{3}{1}{J}
										\leftAttbox{5}{3}{1}{M}
										\rightObjbox{6}{3}{1}{J}
										\rightAttbox{6}{3}{1}{M}

									\end{diagram}}

							\end{picture}}			
						
		\end{minipage}
		\begin{minipage}{0.3\textwidth}
			\centering
							{\unitlength 0.6mm
					\begin{picture}(60,75)%
										\put(0,0){%
											\begin{diagram}{60}{75}
												\Node{1}{30}{0}
												\Node{2}{15}{15}
												\Node{3}{45}{15}
												\Node{4}{30}{30}
												\Node{5}{45}{45}
												\Node{6}{60}{30}
												\Node{7}{30}{60}
												\Node{8}{15}{45}
												\Node{9}{0}{30}

												\Edge{1}{2}
												\Edge{1}{3}
												\Edge{3}{4}
												\Edge{2}{4}
												\Edge{5}{4}
												\Edge{5}{6}
												\Edge{6}{3}
												\Edge{9}{8}
												\Edge{7}{8}
												\Edge{6}{7}				  
												\Edge{9}{2}
												\Edge{3}{8}
												\leftObjbox{2}{3}{1}{J}
												\rightObjbox{3}{3}{1}{J}
												\rightObjbox{6}{3}{1}{J}
												\leftObjbox{9}{3}{1}{J}
												
										\end{diagram}}
										
								\end{picture}}		
		\end{minipage}
		\begin{minipage}{0.3\textwidth}
			\centering
						{\unitlength 0.6mm
	\begin{picture}(60,75)%
		\put(0,0){%
			\begin{diagram}{60}{75}
				\Node{1}{30}{0}
				\Node{2}{15}{15}
				\Node{3}{30}{15}
				\Node{4}{45}{15}
				\Node{5}{15}{30}
				\Node{6}{30}{30}
				\Node{7}{45}{30}
				\Node{8}{30}{45}
				\Edge{1}{2}
				\Edge{1}{3}
				\Edge{1}{4}
				\Edge{2}{5}
				\Edge{2}{6}
				\Edge{3}{5}
				\Edge{3}{7}
				\Edge{4}{6}
				\Edge{4}{7}
				\Edge{5}{8}
				\Edge{6}{8}
				\Edge{7}{8}
				%\Numbers\Edge{2}{3}
				\leftAttbox{5}{3}{1}{M}
				\rightAttbox{6}{3}{1}{M}
				\rightAttbox{7}{3}{1}{M}
		\end{diagram}}
		
\end{picture}}			
		\end{minipage}
		\caption{ $L ^\mathsf{\partial} $ (left), its down-set $BC$ (middle) and its up-set $BC$ (right).}
		\label{fig-S7dual-completions}
\end{figure}

We conclude this section by observing that, in general, the up-set and the down-set Birkhoff completions of a lattice are not isomorphic. For example, consider Figures~\ref{fig-S7-completions} and~\ref{fig-S7dual-completions}.
However, as the example shows, both types of completions are closely related: The down-set Birkhoff completion of a finite lattice $L$ is isomorphic to the dual of the (up-set) Birkhoff completion of the dual of $L$:

\begin{lemma}
	Let $L$ be a finite lattice. 
	Then ($\mathcal{I}(\mathcal{J}(L),\subseteq) ) \cong (\mathcal{F}(\mathcal{M}(L^\mathsf{\partial}),\supseteq))^\mathsf{\partial}$.
\end{lemma}
\begin{proof}
    
	$	 
	(\mathcal J(L),\subseteq) \cong (\mathcal{M}
	(L^\mathsf{\partial}),\supseteq)\implies \\
	(\mathcal{I}(\mathcal J(L)),\subseteq) 
	\cong (\mathcal{F}(\mathcal{M}(L^\mathsf{\partial}),\subseteq))
	\cong (\mathcal{F}(\mathcal{M}(L^\mathsf{\partial}),\supseteq))^\mathsf{\partial}$.
\qed

\end{proof}

Note that the up-set and the down-set Birkhoff completions of a lattice $L$ are of course isomorphic whenever $L$ is distributive --- because in that case both are also isomorphic to $L$ itself. By duality, they are also isomorphic when $L$ is isomorphic to its dual, as shown in Figure~\ref{fig-M3}.

\section{The Birkhoff completion of a formal context}
  \label{sec-birkhoff-context}

We will now transfer the construction introduced in the last section to the case when we do not start with a lattice, but with a formal context. 

We start by recalling from Section \ref{sec-fixing} that for a finite lattice $L$ its Birkhoff completion can be understood as the concept lattice $\uB(\mathcal{M}(L),\mathcal{M}(L),\not\geq)$.
In the sequel, it will be convenient to differentiate, for any element of $\mathcal{M}(L)$, if it is considered in a specific situation as element of the object set or of the attribute set of the context $(\mathcal{M}(L), \mathcal{M}(L),\not\geq)$. Therefore, we will frequently write this context as $(\overline{\mathcal{M}(L)}, \mathcal{M}(L),\not\geq)$, where $\overline{\mathcal{M}(L)}$ contains a copy $\overline{m}$ of each $m\in\mathcal{M}(L)$. The motivation for this notation of negating $m$ is that, in the Birkhoff completion, for $m\in L$, the concept generated by the object $\overline{m}$ is the largest concept that does not have the attribute $m$ in its intent.

We can now define the Birkhoff completion of a formal context $\K:=(G,M,I)$ as extension of the set $G$ of objects by the additional objects $\overline{m}\in\overline{\mathcal{M}(M)}$. The incidence relation $I$ will be extended by $\not\geq$ for the new objects.  
It will be convenient, for $m,n\in M$, to write $m\geq_\K n$ iff $\{m\}'\supseteq \{n\}'$ in $\K$ (which is equivalent to $n$ implying $m$ in the context $\K$ ($\K\models n\to m$) and to $m\in \{n\}''$).

\begin{definition}
	\label{def-BC}
	Let $\K:=(G,M,I)$ be a finite formal context. 

	The \emph{Birkhoff completion of $\K$} is 
	\[BC(\K):=(G\cup\overline{\mathcal{M}(M)}, M, I\cup \{(\overline{m}, n)\in \overline{\mathcal{M}(M)}\times M \mid m\not\geq_\K n\})\]
	with $\overline{\mathcal{M}(M)}:= \{\overline{m} \mid m\in \mathcal{M}(M)\}$ where $\mathcal{M}(M)$ is the set of all meet-irreducible attributes of $\K$.
\end{definition}
The following theorem shows that this is indeed a reasonable definition. 
\begin{theorem}
	\label{theorem-context-BC}
	Let $\K$ be a finite formal context. Then $\uB(BC(\K))\cong BC(\uB(\K))$.
\end{theorem}
\begin{proof}
	Let $\K:=(G,M,I)$ be a finite formal context. % and $L:=\uB(\K)$.  
	\begin{align}
				  & BC(\uB(\K)) = \uB(\mathcal{M}(\uB(\K)),\mathcal{M}(\uB(\K)),\not\geq)  \\
%				=\enspace & \uB(\mathcal{M}(L), \mathcal{M}(L), \not\geq)\\
	\cong\enspace & \uB\left(\enspace\enspace\overline{\mathcal{M}(M)}\enspace\enspace, \mathcal{M}(M),  \{(\overline{m},n) \in \overline{\mathcal{M}(M)}\times \mathcal{M}(M) \mid m\not\geq_\K n\}\right)
		\\
	\cong\enspace & \uB\left(\enspace\enspace\overline{\mathcal{M}(M)}\enspace\enspace, \quad M\quad, \enspace\{(\overline{m},n) \in \overline{\mathcal{M}(M)}\times \enspace M \enspace \mid m\not\geq_\K n\}\right) \label{line3}\\
	\cong\enspace & \uB\left(G\cup\overline{\mathcal{M}(M)}, \quad M\quad , I\cup \{(\overline{m}, n)\in \overline{\mathcal{M}(M)}\times M \mid m\not\geq_\K n \}\right) \label{line4}\\
	=\enspace & \uB\left(BC(\K)\right)
		\end{align}
		
	(1)$\,\cong\,$(2): The meet-irreducible elements of $\uB(\K)$ are exactly the attribute concepts of the attributes in $\mathcal{M}(M)$, and  $\mu(m)\not\geq \mu(n) \iff m\not\geq_\K n m$.
	
	(2)$\,\cong\,$(3): We show that any attribute $n\in M\setminus \mathcal{M}(M)$ is reducible in the context of line~\ref{line3}: Let $B:=\{m\in\mathcal{M}(M) \mid m\geq_\K n\}$. We show that $B^{\dag} = n^{\dag}$ holds in the context of line~\ref{line3}, where $\cdot^{\dag}$ is the derivation operator in this context: First we observe that $B$ is an order filter in $(\mathcal{M}(M), \leq)$, its complement $C:=\mathcal{M}(M)\setminus B$ is  an order ideal in $(\mathcal{M}(M), \leq)$, and it holds $C=B^\dag$. On the other hand, we also have $n^\dag = \{\overline m\in\mathcal{M}(M) \mid m\not\geq_\K n\} = \mathcal{M}(M)\setminus \{\overline m\in\mathcal{M}(M) \mid m\geq_\K n\} = \mathcal{M}(M) \setminus B = C$.

%	$\overline{m}\in B^\dag 
%	 \iff \forall b\in B\colon \overline{m}\in b^\dag
%	 \iff \forall b\in B\colon m\notin b'' 
%	 \iff \forall b\in B\colon b'\not\subseteq m'  
%\iff\todo{...} 
%	 \iff \neg \exists b\in B\colon m'\subseteq b' 
%	 \iff m' \notin \bigcap_{b\in B} b'
%\iff \todo{...} 
%     \iff B' \not\subseteq m'
%	 \iff n' \not\subseteq m' 
%	 \iff m\notin n'' 
%	 \iff \overline{m}\in n^\dag$. 
%	 
%	 \iff \forall b\in B\colon b\notin m'' 
%\iff \forall b\in B\colon m'\not\subseteq b'  
%\iff \neg \exists b\in B\colon m'\subseteq b' 
%\iff m' \notin \bigcap_{b\in B} b'
%\iff m' \notin n' 
%\iff n\notin m'' 
%\iff \overline{m}\in n^\dag$. 
			
	(3)$\,\cong\,$(4): We use $\cdot'$ as symbol for the derivation operator in $\K$ and $\cdot^!$ for the derivation operator in $BC(\K)$ (which is the context shown in line~\ref{line4}). Let $g\in G$. We show that $g$ is reducible in $BC(\K)$ by showing 
	$g^! = A^!$ with 
	$\overline{A}:= \{\overline{m}\in\overline{\mathcal{M}(M)} \mid \neg (g,m)\in I \}$: 
	$A^! = \bigcap_{\overline{m}\in \overline{A}} \overline{m}^! 
	     = \left\{ n\in M \mid \forall \overline{m} \in \overline{A}\colon m\not\geq_\K n \right\}
	     = \left\{ n\in M \mid \forall m\in  \mathcal{M}(M)\colon \neg (g,m)\in I \implies m\not\geq_\K n \right\}
	     = \{ n\in M \mid \forall m\in  \mathcal{M}(M)\colon $ $m\geq_\K n \implies (g,m)\in I \}
         = \{n\in M \mid (g,n)\in I \}
         = g' = g^!$\enspace.

	(4)$\,\cong\,$(5): By definition.
\qed
\end{proof}

\begin{example}
	The following example illustrates the construction. 
Figure~\ref{fig-british-isles} shows the complex usage of terms of the administrative geography of the United Kingdom \cite{wikipedia24britishisles}, resulting from a long (and often bloody) historical process. The concept lattice in Figure~\ref{fig-british-isles-lattice} is based on the formal context which has the five countries (England, Scotland, Wales, Northern Ireland and Ireland) and the three crown dependencies (Isle of Man, Jersey and Guernsey) in the British Isles as objects, and the geographical features (marked green in Figure~\ref{fig-british-isles}) `British Isles', `Ireland (Island)' and `Great Britain' as well as the legal distinctions (blue) `British Islands' and `United Kingdom' as attributes. We added the geographical feature `Channel Islands' as further attribute, as it is discussed as such in \cite{wikipedia24britishisles}.  

We note that the lattice is not completely distributive, as three of its sublattices are isomporphic to $N_5$. 
The up-set Birkhoff completion of the concept lattice is shown in Figure~\ref{fig-british-isles-upsetBC}. Following the construction of Definition~\ref{def-BC}, it contains an additional, artificial (and potentially interesting) object for each of the five attributes. 
The down-set Birkhoff completion is shown in Figure~\ref{fig-british-isles-downsetBC}. It contains an additional, artificial (and potentially interesting) attribute for each of the eight objects. In Section~\ref{sec-analysis}, we will discuss what we can learn from these completions about the administrative geography of the UK. 
\end{example}

\begin{figure}
 	\begin{center}
 		\includegraphics[width=7cm]{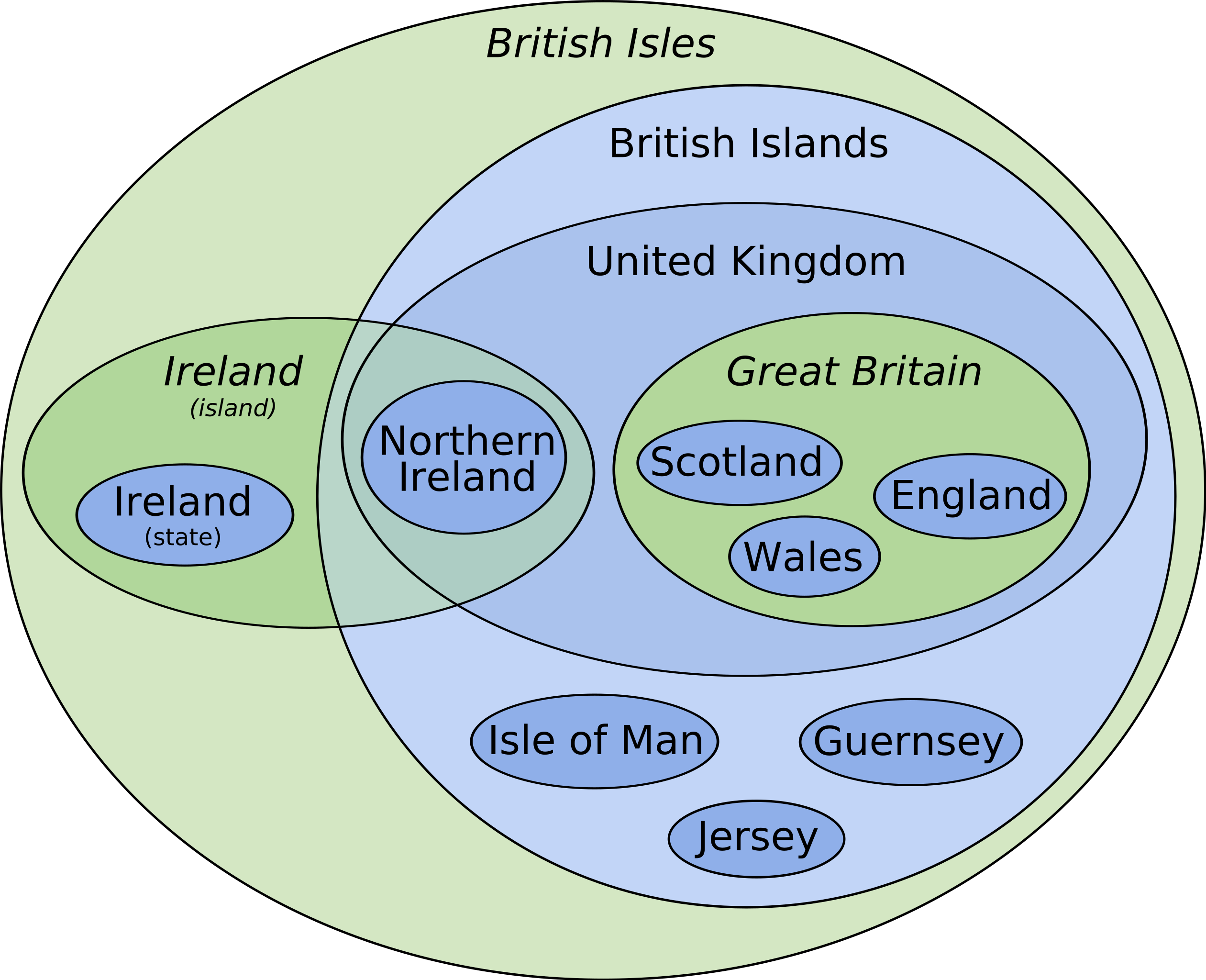}
 	\end{center}
	\caption{Administrative geography of the United Kingdom  \cite{wikipedia24britishisleseuler}.}
	\label{fig-british-isles}
\end{figure}

\begin{figure}
	\begin{center}
          \colorlet{mivertexcolor}{white}
\colorlet{jivertexcolor}{white}
\colorlet{vertexcolor}{white}
\colorlet{bordercolor}{black}
\colorlet{linecolor}{gray}
% parameter corresponds to the used valuation function and can be addressed by #1
\tikzset{vertexbase/.style 2 args={semithick, shape=circle, inner sep=2pt, outer sep=0pt, draw=bordercolor},%
  vertex/.style 2 args={vertexbase={#1}{}, fill=vertexcolor!45},%
  mivertex/.style 2 args={vertexbase={#1}{}, fill=mivertexcolor!45},%
  jivertex/.style 2 args={vertexbase={#1}{}, fill=jivertexcolor!45},%
  divertex/.style 2 args={vertexbase={#1}{}, top color=mivertexcolor!45, bottom color=jivertexcolor!45},%
  conn/.style={-, thick, color=linecolor}%
}
\begin{tikzpicture}[scale=2]
  \begin{scope} %for scaling and the like
    \begin{scope} %draw vertices
      \foreach \nodename/\nodetype/\param/\xpos/\ypos in {%
        0/vertex//0/0,
        1/divertex//0/1,
        2/jivertex//-1/1,
        3/divertex//1/1,
        4/mivertex//-0.5/1.5,
        5/divertex//-2/2,
        6/mivertex//0/2,
        7/vertex//-1/3
      } \node[\nodetype={\param}{}] (\nodename) at (\xpos, \ypos) {};
    \end{scope}
    \begin{scope} %draw connections
      \path (3) edge[conn] (6);
      \path (1) edge[conn] (4);
      \path (2) edge[conn] (5);
      \path (2) edge[conn] (4);
      \path (5) edge[conn] (7);
      \path (0) edge[conn] (3);
      \path (0) edge[conn] (1);
      \path (0) edge[conn] (2);
      \path (6) edge[conn] (7);
      \path (4) edge[conn] (6);
    \end{scope}
    \begin{scope} %add labels
      \foreach \nodename/\labelpos/\labelopts/\labelcontent in {%
        1/above//{Great Britain},
        3/above right//{Channel Islands},
        4/above//{United Kingdom},
        5/above//{Ireland (Island)},
        6/above//{British Islands},
%        7/above//{British Isles},
        1/below//{\parbox{2.5cm}{\baselineskip=0pt Wales, Scotland,\\ England}},
        2/below left//{Northern Ireland},
        3/below right//{\parbox{2.5cm}{\baselineskip=0pt Jersey, Isle of\\ Man, Guernsey}},
        5/below//{Ireland (State)}
      } \coordinate[label={[\labelopts]\labelpos:{\labelcontent}}](c) at (\nodename);
    \end{scope}
  \end{scope}
\end{tikzpicture}
%		\todo{Bild figs/British\_Isles\_lattice} %	\includegraphics[width=7cm]{figs/British_Isles_lattice}
%			\includegraphics[width=7cm]{figs/handdrawn-british-isles-lattice.jpeg}
	\end{center}
	\caption{Concept lattice about the administrative geography of the United Kingdom.}
	\label{fig-british-isles-lattice}
\end{figure}
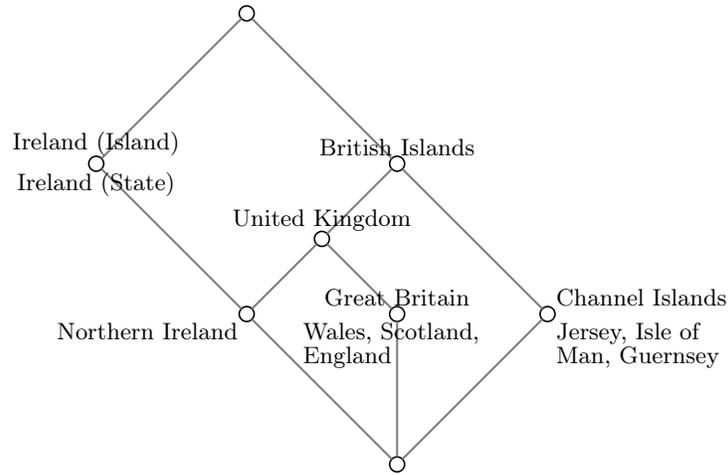

\begin{figure}
	\begin{center}
          \colorlet{mivertexcolor}{white}
\colorlet{jivertexcolor}{white}
\colorlet{vertexcolor}{white}
\colorlet{bordercolor}{black}
\colorlet{linecolor}{gray}
% parameter corresponds to the used valuation function and can be addressed by #1
\tikzset{vertexbase/.style 2 args={semithick, shape=circle, inner sep=2pt, outer sep=0pt, draw=bordercolor},%
  vertex/.style 2 args={vertexbase={#1}{}, fill=vertexcolor!45},%
  mivertex/.style 2 args={vertexbase={#1}{}, fill=black!60},%
  jivertex/.style 2 args={vertexbase={#1}{}, fill=jivertexcolor!45},%
  divertex/.style 2 args={vertexbase={#1}{}, top color=mivertexcolor!45, bottom color=jivertexcolor!45},%
  conn/.style={-, thick, color=linecolor}%
}
\begin{tikzpicture}[scale=0.4]
  \begin{scope} %for scaling and the like
    \begin{scope} %draw vertices
      \foreach \nodename/\nodetype/\param/\xpos/\ypos in {%
        0/mivertex//-1./4,
        1/vertex//-1/8,
        2/vertex//4/10,
        3/vertex//-6/10,
        4/vertex//-1/12,
        5/vertex//4/14,
        6/mivertex//-6/14,
        7/mivertex//-1.0/15.0,
        8/mivertex//4/18,
        9/vertex//-6/18,
        10/mivertex//-1.0/19.0,
        11/mivertex//-6/22,
        12/mivertex//-1.0/23.0,
        13/mivertex//-1/27
      } \node[\nodetype={\param}{}] (\nodename) at (\xpos, \ypos) {};
    \end{scope}
    \begin{scope} %draw connections
      \path (12) edge[conn] (13);
      \path (7) edge[conn] (10);
      \path (3) edge[conn] (7);
      \path (3) edge[conn] (6);
      \path (6) edge[conn] (10);
      \path (6) edge[conn] (9);
      \path (10) edge[conn] (12);
      \path (9) edge[conn] (12);
      \path (9) edge[conn] (11);
      \path (0) edge[conn] (3);
      \path (0) edge[conn] (2);
      \path (0) edge[conn] (1);
      \path (8) edge[conn] (12);
      \path (5) edge[conn] (10);
      \path (5) edge[conn] (8);
      \path (11) edge[conn] (13);
      \path (2) edge[conn] (7);
      \path (2) edge[conn] (5);
      \path (1) edge[conn] (6);
      \path (1) edge[conn] (5);
      \path (1) edge[conn] (4);
      \path (4) edge[conn] (9);
      \path (4) edge[conn] (8);
    \end{scope}
    \begin{scope} %add labels
      \foreach \nodename/\labelpos/\labelopts/\labelcontent in {%
        7/above//{Great Britain},
        8/above right//{\parbox{1cm}{\baselineskip=0pt Channel\\ Islands}},
        10/above//{United Kingdom},
        11/above left//{Ireland (Island)},
        12/above//{British Islands},
%        13/above//{British Isles},
        1/below//{$\overline{\text{Great Britain}}$},
        2/below right//{\parbox{1cm}{\baselineskip=0pt $\overline{\text{Ireland}}$\\ (Island)}},
        3/below left//{\parbox{1cm}{\baselineskip=0pt $\overline{\text{Channel}}$\\ Islands}},
        4/below//{$\overline{\text{United Kingdom}}$},
        6/below left//{\parbox{1cm}{\baselineskip=0pt Northern\\ Ireland}},
        7/below//{\parbox{2.5cm}{\baselineskip=0pt Wales, Scotland,\\ England}},
        8/below right//{\parbox{2.2cm}{\baselineskip=0pt Jersey, Isle of Man, Guernsey}},
        11/below left//{\parbox{2.2cm}{\baselineskip=0pt Ireland (State),\\$\overline{\text{British Islands}}$}}
      } \coordinate[label={[\labelopts]\labelpos:{\labelcontent}}](c) at (\nodename);
    \end{scope}
  \end{scope}
\end{tikzpicture}
	\end{center}
	\caption{Up-set Birkhoff completion of the concept lattice in Figure~\ref{fig-british-isles-lattice}.}
	\label{fig-british-isles-upsetBC}
\end{figure}

\begin{figure}
	\begin{center}
          \colorlet{mivertexcolor}{white}
\colorlet{jivertexcolor}{white}
\colorlet{vertexcolor}{white}
\colorlet{bordercolor}{black!80}
\colorlet{linecolor}{gray}
% parameter corresponds to the used valuation function and can be addressed by #1
\tikzset{vertexbase/.style 2 args={semithick, shape=circle, inner sep=2pt, outer sep=0pt, draw=bordercolor},%
  vertex/.style 2 args={vertexbase={#1}{}, fill=vertexcolor!45},%
  mivertex/.style 2 args={vertexbase={#1}{}, fill=black!60},%
  jivertex/.style 2 args={vertexbase={#1}{}, fill=jivertexcolor!45},%
  divertex/.style 2 args={vertexbase={#1}{}, top color=mivertexcolor!45, bottom color=jivertexcolor!45},%
  conn/.style={-, thick, color=linecolor}%
}
\begin{tikzpicture}
  \begin{scope} %for scaling and the like
    \begin{scope} %draw vertices
      \foreach \nodename/\nodetype/\param/\xpos/\ypos in {%
        0/mivertex//1/12,
        1/mivertex//-1/14,
        2/mivertex//1/14,
        3/mivertex//3/14,
        4/vertex//-1/16,
        5/mivertex//-3/16,
        6/mivertex//1/16,
        7/vertex//3/16,
        8/vertex//-1/18,
        9/vertex//-3/18,
        10/mivertex//1/18,
        11/mivertex//-1/20
      } \node[\nodetype={\param}{}] (\nodename) at (\xpos, \ypos) {};
    \end{scope}
    \begin{scope} %draw connections
      \path (6) edge[conn] (10);
      \path (6) edge[conn] (8);
      \path (7) edge[conn] (10);
      \path (3) edge[conn] (6);
      \path (3) edge[conn] (7);
      \path (10) edge[conn] (11);
      \path (9) edge[conn] (11);
      \path (1) edge[conn] (6);
      \path (1) edge[conn] (4);
      \path (1) edge[conn] (5);
      \path (8) edge[conn] (11);
      \path (4) edge[conn] (10);
      \path (4) edge[conn] (9);
      \path (0) edge[conn] (3);
      \path (0) edge[conn] (1);
      \path (0) edge[conn] (2);
      \path (5) edge[conn] (9);
      \path (5) edge[conn] (8);
      \path (2) edge[conn] (7);
      \path (2) edge[conn] (4);
    \end{scope}
    \begin{scope} %add labels
      \foreach \nodename/\labelpos/\labelopts/\labelcontent in {%
        2/above//{Channel Islands},
        3/above right//{Great Britain},
        5/above//{Ireland (Island)},
        6/above//{United Kingdom},
        7/above right//{$\overline{\text{Northern Ireland}}$},
        8/above//{\parbox{1.7cm}{$\overline{\text{Jersey}}$,$\overline{\text{Guernsey}}$,\\ $\overline{\text{Isle of Man}}$}},
        9/above//{\parbox{4cm}{$\overline{\text{Wales}}$,$\overline{\text{Scotland}}$\\ $\overline{\text{England}}$}},
        10/above right//{\parbox{2.2cm}{\baselineskip=4pt British Islands,\\$\overline{\text{Ireland (State)}}$} },
        1/below left//{Northern Ireland},
        2/below//{\parbox{2.5cm}{\baselineskip=0pt Jersey, Guernsey,\\ Isle of Man}},
        3/below right//{\parbox{2.5cm}{\baselineskip=0pt Wales, Scotland,\\ England}},
        5/below//{Ireland (State)}
      } \coordinate[label={[\labelopts]\labelpos:{\labelcontent}}](c) at (\nodename);
    \end{scope}
  \end{scope}
\end{tikzpicture}
	\end{center}
	\caption{Down-set Birkhoff completion of the concept lattice in Figure~\ref{fig-british-isles-lattice}.}
	\label{fig-british-isles-downsetBC}
\end{figure}

%%% Local Variables:
%%% mode: latex
%%% TeX-master: "2024"
%%% End:

\section{An implicational approach to the Birkhoff completion}
	\label{sec-implications}
	
It is long known that a finite concept lattice is distribtive if and only if its canonical direct basis contains only implications that have one-element premises. After having recalled the basic notions, we will show in this section, that this leads in a straightforward manner to another equivalent approach to defining the Birkhoff completion.

%In this section, we extend the relation of distributivity and the
%\textbf{canonical direct basis} to a concept lattice and its Birkhoff
%completion. For this, we first recall this bases whose premises are
%given by the set of all proper premises. We recall their definition in
%terms of the set of objects.

\begin{definition}[Proper Premise \cite{fca-book}]
  Let $(G,M,I)$ be a context. An attribute set $A\subseteq M$ is a
  \textbf{proper premise} if
  \[ A^{\bullet}\coloneqq A'' \setminus \big(A \cup \bigcup_{n\in A} (A\ \{n\})''\big)\]
  is not empty.
\end{definition}

\begin{definition}
	The canonical direct basis of a finite context $\K$ is defined as the set $\mathcal{CDB}(\K)$ of all
implications $A\to A^{\bullet}$, where $A$ is a proper premise. 
\end{definition}

Based
on the following theorem, we can infer that there is correspondence
between distributive lattices and those that are isomorphic to concept
lattices whose contexts proper premises are singletons.

\begin{theorem}[Theorem 41\,(8) of \cite{fca-book}]
  \label{TH41}
  Let $\K$ be a finite, reduced context. Then $\uB(\K)$ is distributive iff every proper premise of $(G,M,I)$ is a singleton.
\end{theorem}

Thus the reduced contexts of concept lattices that are not distributive have
proper premises that are of size 2 or larger. This motivates the following definition. 

\begin{definition}
	Let $\K$ be a finite, reduced context and $A\to A^\bullet$ an implication of its canonical direct basis. We call $A\to A^\bullet$ a  \emph{distributive implication} if $|A|=1$ and a \emph{non-distributive implication} if $|A|\geq 2$.
\end{definition}

With
\cref{theorem-context-BC} (1 $\cong$ 4) we showed that the Birkhoff
completion of a concept lattice can be achieved by computing a finer
closure system $\Int(\BC(\context))$ of $\Int(\context)$ on the set
$M$, i.e., $\Int(\context)\subseteq \Int(\BC(\context))$. With the
following proposition, we study this completion in terms of
implications. For that, we make use of the dualism
between closure systems and their implicational theories, i.e.,
$\Int(\context)\subseteq \Int(\BC(\context))$ implies for the set of
all attribute implications $\Th(\BC(\context))\subseteq \Th(\context)$
\cite{lattice_of_closure_systems}.

\begin{lemma}
	\label{lemma-proper-premises}
	Let $\K$ be a finite, reduced context, $\mathcal{A}$ the set of its proper premises, and $BC(\mathcal{A})\coloneqq \{A\in \mathcal{A} \mid |A|=1 \}$. For all $A\in BC(\mathcal{A})$ holds $A^\bullet = A''\setminus A$ independent of being computed in $\K$ or in $BC(\K)$. Furthermore $BC(\mathcal{A}) $ is the set of all proper premises of $BC(\K)$. 
\end{lemma}
\begin{proof}
	As in the proof of Theorem~\ref{theorem-context-BC}, we use $\cdot'$ as symbol for the derivation operator in $\K$ and $\cdot^!$ for the derivation operator in $BC(\K)$. 
	We start by observing that $\BC(\context)$ is attribute reduced since $\context$ is attribute reduced and
	$\Int(\context)\subseteq \Int(\BC(\context))$. As both $\K$ and $BC(\K)$ are attribute reduced, we have in particular $\emptyset''=\emptyset=\emptyset^{!!}$.
	Let now $A=\{n\}\in BC(\mathcal{A})$. 
	With
	 \begin{multline*}
	 	A^{!!}=(A^!)^!
	 = (A'\cup\{\overline{m}\in\overline{\mathcal{M}(M)} \mid m \not\geq_\K n \})^!
	 = (A')^!\cap\{\overline{m}\in\overline{\mathcal{M}(M)} \mid m \not\geq_\K n \}^!   \\ 
	 =  A''  \cap \{k\in M \mid  m \not\geq_\K k \textrm{ for all } m \textrm{ with } m  \not\geq_\K n \}
	 =  A''  \cap \{k\in M \mid k\geq_{\context} n\}
	 =  A''
	 \end{multline*}
	 we obtain 
	 \begin{multline*}
	 A^\bullet\enspace (\textrm{computed in } BC(\K)) 
	 = A^{!!}\setminus(A\cup\bigcup_{n\in A} A \setminus\{n\}^{!!}) \\
		= A^{!!}\setminus(A\cup\bigcup_{n\in A} \emptyset^{!!})
		= A^{!!}\setminus A 
	   = \qquad A''\setminus A \qquad
		= A''\setminus(A\cup\bigcup_{n\in A} \emptyset'')\\
	 = A''\setminus(A\cup\bigcup_{n\in A} A \setminus\{n\}'') 
	 = A^\bullet \enspace (\textrm{computed in } \K)\enspace.
	 \end{multline*}
	 It follows immediately that $A$ is a proper premise of $BC(\K)$ because it is a proper premise of $\K$ and thus $A^\bullet$ is not empty. 
	 
	 It remains to show that all proper premises of $BC(\K)$ are contained in $BC(\mathcal{A})$. Let $A$ be a proper premise of $BC(\K)$.
	 We can conclude from
	 $\Int(\context)\subseteq \Int(\BC(\context))$ that
	 $A^{I_{\context}I_{\context}}\neq A$. Moreover, since $\context$
	 is reduced, we can infer that
	 $\emptyset\in \Int(\context)$. Thus $A^{\bullet}$ (computed in
	 $\context$) is not empty and $A$ is a proper premise of $\context$.
	 Additionally, we have $|A|=1$ because of the distributivity of $\uB(BC(\K))$. 
	 \qed\end{proof}

\begin{corollary}
	\label{corollary-BC-CDB}
	For a finite reduced context  $\K$ is $BC(\mathcal{CDB}(\K))=\mathcal{CDB}(BC(\K))$  with 
		\[BC(\mathcal{CDB}(\K)) \coloneqq \{A\to A^{\bullet} \mid A \to A^{\bullet} \in \mathcal{CDB}(\K) \text{ and } \abs{A}=1\}\enspace.\]
\end{corollary}

Note that for all proper premises $A$ of $\K$ with $|A|\geq 2$, $A$ is
not a proper premise of $BC(\K)$, and $A^\bullet$ evaluated in
$BC(\K)$ will be a strict subset of $A^\bullet$ evaluated in $\K$.

With \Cref{lemma-proper-premises} and \cref{corollary-BC-CDB}, we have not only shown a relation between a
concept lattice, its Birkhoff completion and their canonical direct bases, but we can also derive a
procedure for the computation of the Birkhoff completion of the concept lattice. That is, compute the closure system of
attribute sets that are closed with respect to the implicational theory $BC(\mathcal{CDB}(\K))$.

The condition that $\context$ has to be reduced seems to be
a constraint for the application of this proposition. However for a
context that is not reduced, we can first reduce it, then compute
the completion with respect to proper premises and lastly re-add the
reducible attributes to the completion.

%\begin{example}
	The canonical direct basis of the concept lattice (Figure~\ref{fig-british-isles}) about the administrative geography of the UK is given in Figure \ref{fig-implications}. According to \Cref{lemma-proper-premises}, the canonical direct base of the Birkhoff completion of the  would thus consist of implications (1) to (3) --- which can easily be verified in Figure~\ref{fig-british-isles-upsetBC}. 
	
%		The implications that are associated to the three $N_5$-type sublattices of Figure \ref{fig-british-isles-lattice} are $\{\textrm{Ireland (Island), British Islands}\}\to\{\textrm{UK}\}$, $\{\textrm{UK, Channel Islands}\}\to\{\textrm{GB}\}$ and $\{\textrm{Ireland (Island), UK}\}\to\{\textrm{Channel Islands}\}$. The premise of the first implication is fulfilled by Northern Island, while the premises of the second and the third are inconsistent, \ie there are no objects having any of those two attribute combinations.
%\end{example}

\begin{figure}
	\setcounter{equation}{0}
\begin{align}
		\{\textrm{UK}\} & \to \{\textrm{British Islands}\}\\
		\{\textrm{GB}\} & \to  \{\textrm{UK, British Islands}\} \\
		\{\textrm{Channel Islands}\} & \to   \{\textrm{British Islands}\} \\[1ex]
		\{\textrm{Ireland (Island), British Islands}\} & \to   \{\textrm{UK}\} \\   %[1ex]
		\{\textrm{UK, Channel Islands}\} & \to   \{\textrm{Ireland (Island), British Islands, GB}\} \\
		\{\textrm{Channel Islands, GB}\} & \to   \{\textrm{UK, Ireland (Island), British Islands}\} \\
		\{\textrm{Ireland (Island), Channel Islands}\} & \to   \{\textrm{UK, British Islands, GB}\} \\
		\{\textrm{Ireland (Island), GB}\} & \to   \{\textrm{UK, British Islands, Channel Islands}\} 
\end{align}

	\caption{Canonical direct basis of the administrative geography of the UK.}
		\label{fig-implications}
\end{figure}

%%% Local Variables:
%%% mode: latex
%%% TeX-master: "2024"
%%% End:

\forlongversion{\input{truncated-BC}}
\section{Application: analysis of the administrative geography of the UK}
	\label{sec-analysis}
	
We will now study by an example how the Birkhoff completion of the context interacts with the completion by implications; and what can be learnt from the Birkhoff completion. To this end, we consider again the administrative geopgraphy of the UK. 

In the top left of the up-set $BC$ in Figure~\ref{fig-british-isles-upsetBC}, we can see that $\overline{\textrm{British Islands}}$ is generating the same concept as Ireland (State). This indicates that the legal distinction `British Islands' is defined by exclusion -- it contains all countries\footnote{For sake of better readability, we will, in the sequel, subsume the crown dependencies under the term `country'.} of the British Isles except the state of Ireland. 

The other four new objects are all found at the lower part of the concept lattice, and all generate new concepts. These objects invalidate the implications with two-element premises of Figure~\ref{fig-implications}: 
$\{\textrm{Ireland (Island), British Islands}\}\to\{\textrm{UK}\}$ is invalidated by $\overline{\textrm{UK}}$, $\{\textrm{UK, Channel Islands}\}\to\{\textrm{GB}\}$ by $\overline{\textrm{GB}}$, and \linebreak $\{\textrm{Ireland (Island), GB}\}\to\{\textrm{Channel Islands}\}$ by $\overline{\textrm{Channel Islands}}$. 

Let us now assume for a moment, that geographical features are given by nature and not subject to modification, while legal distinctions might change over the years. We would thus assume that Ireland (Island), GB and the Channel Islands will always remain geographically disjoint. A new country having all attributes of $\overline{\textrm{Channel Islands}}$ will therefore never exist.
% (In terms of implications that would yield $\{\textrm{Ireland (Island), GB } \} \to\{\bot\}$.) 
On the other hand, in the future there may arise new countries on the British Isles that are affiliated to different combinations of legal distinctions ---- or existing countries might change their affiliations. Therefore we cannot exclude the potential existence of countries having the attributes of $\overline{\textrm{UK}}$, of $\overline{\textrm{GB}}$ or of $\overline{\textrm{Ireland(Island)}}$. 

\forlongversion{Overall, the resulting truncated BC is displayed in Figure~\ref{fig-british-isles-truncatedBC}.}

We conclude this study of UK's administrative geography by a look at the down-set Birkhoff completion in Figure~\ref{fig-british-isles-downsetBC}. First, we observe that all implications of the canonical direct base (Figure~\ref{fig-implications}) still hold. Note that this is not contradicting Theorem~\ref{TH41} even though the down-set completion is distributive, because the original attributes (except `British Islands') are not reduced any more.  

As in the up-set completion we can observe that `British Islands' is the complement of `Ireland (State)', because $\overline{\mbox{Ireland (State)}}$ is attached to the same concept as `British Islands'. The seven other new attibutes generate three new concepts. These concepts can be understood in two ways. The first is as disjunction of the attributes below. The top left concept, for instance, could be understood as `Ireland (Island) \texttt{or} Channel Islands'. The second way is to consider them as negation, but note that we are now negating objects. While one is tempted to consider   $\overline{\mbox{Wales}}$,  $\overline{\mbox{Scotland}}$ $\overline{\mbox{England)}}$ at the top left as `$\neg$ Great Britain', the attribute  $\overline{\mbox{Northern Island}}$ shows that it is not that simple: Because of the way it was constructed this attribute is incident with all objects that are neither `Northern Island' nor other objects (such as `Ireland (State)') which are implied (as object implication) by `Northern Island'. Either way, we can understand these newly generated concepts as proposals for new concepts which are waiting for a denomination.

%%%---------------------
\forlongversion{
	
We conclude this section with a classical example~\cite[p. 17ff]{fca-book}): The concept lattice displayed in Figure~\ref{fig-living-beings} was used for the planning of an Hungarian educational film entitled ``Living Beings and Water''. Its aim was to put features of animals and plants that were relevant to the topic in a suitable order and to find illustrating objects. 

If we were interested in the full implicational logic of this domain together with a set of living beings that are representative for the whole data set, performing an some kind of exploration \cite{ganter2016conceptual} --- e.\,g., attribute exploration \cite{ganter1999attribute} or distributive concept exploration \cite{stumme1995knowledge} --- would be appropriate. If either the required effort or the expected outcome are considered too large --- certainly the educational film shall not replace a whole study of biology ---  one might want to start with the context already given and check only for some local improvements. For this purpose, \todo{...}

The canonical direct base consists of 18 implications. Four of them are distributive, eg. $\{$suckles its offspring$\}\to\{$has limbs, lives on water, can move$\}$, the other 14 have premise size 2. For 12 of these 14 non-distributive implications, the infimum of the premise is the bottom element of the lattice, which means that none of the existing objects have this attribute combination. Most of these attribute combinations will not occur for any living being (e\,g., `suckles its offspring' and `needs chlorophyll') and can easily be added to the set of inconsistencies. Only one combination is debatable: is there some animal that suckles its offspring and lives in water? 

Its truncated Birkhoff completion is displayed in Figure~\ref{fig-living-beings-truncatedBC}. The truncation was performed such that all premises of size 2 or larger where checked if they are logically inconsistent. All these premises were included in the set of inconsistencies. Then all implications with premise size 2 or larger were removed. 

\begin{figure}[t]
	\begin{center}
    \includegraphics[height=6cm]{figs/fische_hand.tikz}
	\end{center}
	\caption{Concept lattice for the educational film ``Living Beings and Water''~(reproduced from \protect\cite[p. 24]{fca-book}).}
	\label{fig-living-beings}
%	\caption{Living Beings and 		Water ---  context from removing proper premises of size two and larger that do not imply all attributes.}
\end{figure}

\begin{figure}[t]
	\begin{center}
		\colorlet{mivertexcolor}{white}
\colorlet{jivertexcolor}{white}
\colorlet{vertexcolor}{white}
\colorlet{bordercolor}{black}
\colorlet{linecolor}{gray}
% parameter corresponds to the used valuation function and can be addressed by #1
\tikzset{vertexbase/.style 2 args={semithick, shape=circle, inner sep=2pt, outer sep=0pt, draw=bordercolor},%
  vertex/.style 2 args={vertexbase={#1}{}, fill=vertexcolor!45},%
  mivertex/.style 2 args={vertexbase={#1}{}, fill=mivertexcolor!45},%
  jivertex/.style 2 args={vertexbase={#1}{}, fill=jivertexcolor!45},%
  divertex/.style 2 args={vertexbase={#1}{}, top color=mivertexcolor!45, bottom color=jivertexcolor!45},%
  conn/.style={-, thick, color=linecolor}%
}
\begin{tikzpicture}[scale=0.3,font=\footnotesize]
  \begin{scope} %for scaling and the like
    \begin{scope} %draw vertices
      \foreach \nodename/\nodetype/\param/\xpos/\ypos in {%
        0/vertex//6.996829574834679/-1.051083022120089,
        1/jivertex//-5.002966101695072/4.959322033898305,
        2/divertex//9.983161088767456/5.0104589198506115,
        3/jivertex//19.069067796610014/5.023728813559321,
        4/divertex//-8.948047615627845/5.063786268313697,
        5/jivertex//14.979237288135435/7.052542372881357,
        6/jivertex//-0.9453389830510037/7.06864406779661,
        7/jivertex//21.178389830508316/8.88813559322034,
        8/vertex//-7.047881355932323/8.920338983050849,
        9/jivertex//17.120762711864252/8.96864406779661,
        10/jivertex//-3.0546610169492787/9.049152542372882,
        11/vertex//16.975847457626962/10.852542372881357,
        12/vertex//12.98262711864391/10.884745762711866,
        13/vertex//-2.9902542372882497/10.949152542372882,
        14/vertex//7.024999999999844/10.949152542372882,
        15/vertex//0.9224576271185505/10.98135593220339,
        16/mivertex//19.069067796610014/12.849152542372883,
        17/mivertex//-4.890254237288205/12.897457627118646,
        18/mivertex//-0.8326271186441314/14.92627118644068,
        19/mivertex//15.043644067796452/14.942372881355933,
        20/mivertex//5.044491525423677/14.974576271186443,
        21/mivertex//9.037711864406738/15.055084745762713,
        22/vertex//7.0/19.0
      } \node[\nodetype={\param}{}] (\nodename) at (\xpos, \ypos) {};
    \end{scope}
    \begin{scope} %draw connections
      \path (20) edge[conn] (22);
      \path (10) edge[conn] (15);
      \path (8) edge[conn] (13);
      \path (13) edge[conn] (18);
      \path (9) edge[conn] (16);
      \path (21) edge[conn] (22);
      \path (1) edge[conn] (8);
      \path (4) edge[conn] (8);
      \path (3) edge[conn] (7);
      \path (9) edge[conn] (12);
      \path (8) edge[conn] (17);
      \path (14) edge[conn] (20);
      \path (12) edge[conn] (19);
      \path (7) edge[conn] (11);
      \path (0) edge[conn] (4);
      \path (5) edge[conn] (11);
      \path (6) edge[conn] (13);
      \path (12) edge[conn] (20);
      \path (16) edge[conn] (19);
      \path (17) edge[conn] (18);
      \path (0) edge[conn] (3);
      \path (5) edge[conn] (14);
      \path (5) edge[conn] (12);
      \path (15) edge[conn] (18);
      \path (13) edge[conn] (20);
      \path (11) edge[conn] (21);
      \path (6) edge[conn] (14);
      \path (1) edge[conn] (10);
      \path (15) edge[conn] (21);
      \path (6) edge[conn] (15);
      \path (10) edge[conn] (17);
      \path (1) edge[conn] (6);
      \path (11) edge[conn] (19);
      \path (0) edge[conn] (1);
      \path (18) edge[conn] (22);
      \path (7) edge[conn] (16);
      \path (14) edge[conn] (21);
      \path (3) edge[conn] (5);
      \path (3) edge[conn] (9);
      \path (2) edge[conn] (12);
      \path (19) edge[conn] (22);
      \path (0) edge[conn] (2);
    \end{scope}
    \begin{scope} %add labels
      \foreach \nodename/\labelpos/\labelopts/\labelcontent in {%
        2/above//{\parbox{2cm}{\baselineskip=0pt two seed\\ leaf}},
        4/above//{\parbox{2cm}{\baselineskip=0pt suckles its\\ offspring}},
        16/above//{\parbox{1.5cm}{\baselineskip=0pt one seed\\ leaf}},
        17/above//{has limbs},
        18/above//{\parbox{2cm}{\baselineskip=0pt can move\\ around}},
        19/above//{\parbox{1.5cm}{\baselineskip=0pt needs\\ chlorophyll}},
        20/above//{\parbox{1.5cm}{\baselineskip=0pt lives on\\ land}},
        21/above//{\parbox{1cm}{\baselineskip=0pt lives in\\ water}},
        22/above//{needs water to live},
        1/below//{frog}, %\{has limbs, lives on land, lives in water, can move\}
        2/below//{bean}, % \{needs chlorophyll, lives on land, dicotyledon\}
        3/below//{reed}, % \{needs chlorophyll, lives on land, lives in water, monocotyledon\}
        4/below//{dog}, % \{has limbs, breast feeds, lives on land, can move\}
        5/below//{{\color{red} Sphagnum}}, %Torfmoos    \{needs chlorophyll, lives on land, lives in water\}
        6/below//{{\color{red} Caecilian}}, %Blindlurch     \{lives on land, lives in water, can move\}
        7/below//{\parbox{1cm}{\baselineskip=0pt spike\\ weeds}}, %\{needs chlorophyll, lives in water, monocotyledon\}
        9/below//{maize}, %\{needs chlorophyll, lives on land, monocotyledon\}
        10/below//{bream}, %\{has limbs, lives in water, can move\}
        15/below//{leech}
      } \coordinate[label={[\labelopts]\labelpos:{\labelcontent}}](c) at (\nodename);
    \end{scope}
  \end{scope}
\end{tikzpicture}
	\end{center}
	\caption{Truncated BC of the concept lattice of Figure~\protect\ref{fig-living-beings}.}
	\label{fig-living-beings-truncatedBC}
	%	\caption{Living Beings and 		Water ---  context from removing proper premises of size two and larger that do not imply all attributes.}
\end{figure}

}  %%%---------------------

\section{Conclusion}
	\label{sec-conclusion}
	
Motivated by the fact that distributivity in concept lattices arises from the fact that everyday thinking predominantly uses simple implications, we have introduced the Birkhoff completion of a finite lattice $L$ as the minimal distributive lattice in which $L$ can be embeddeded as sub-semilattice.  We have presented a corresponding completion for formal contexts and a corresponding restriction of the implicational basis, and have shown their equivalences. 

Depending on the characteristics of the dataset, either the up-set $BC$ or the down-set $BC$ might be more beneficial. If all objects under consideration are known from the beginning --- like the countries in the British Isles --- the down-set $BC$ might focus our view on potentially new attributes for their description, while the up-set $BC$ is only of interest, if we do not want to exclude the possibility of new objects. 

\forlongversion{In cases where we know only some part of the object set --- like in the living beings example --- the upset BC guides us to structurally relevant objects that we might want to add to the set of explicitly listed objects. }

We have defined the up-set Birkhoff completion of a lattice as standard rather than the down-set completion because of its dualism with the Birkhoff completion of the canonical direct base for the attribute implications. 
The same dualism holds of course also the down-set completion together with the object implications, but the latter are far less promeninent in everyday thinking.

Note that the Birkhoff completions are by no means the only ways to turn a non-distributive lattice into a distributive one. Other obvious ways are the deletion of objects and/or attributes, or the modification of the incidence relation of the formal context. Our future work will study the effects of these options more closely. 
When one is aiming at simplifying a given dataset so that all its implications are distributive, one could choose any of these options. Which choice is more adequate depends on the characteristics of the described objects. 
\forlongversion{For instance, in the living being case it is a natural option to add a new animal or plant to the small set of prototypical examples, while we would not expect that an object is changing its set of attributes.
For the UK example, the situation is quite the opposite: We do not easily see new countries emerging, and the removal of a country would be heavily disputed. However, changing the affiliation of a country to supranational bodies happens from time to time. These examples show that right choice of analysis tool requires at least some understanding of the domain of interest. Nevertheless the variety of approaches might be useful for what-if-analyses.}

\bibliographystyle{splncs04}
\bibliography{paper}
\end{document}